\documentclass[10pt,a4paper]{amsart}
\usepackage[utf8]{inputenc}
\usepackage{amstext, amssymb, amsthm, amsfonts, latexsym,bbm,tikz,sfmath,color,varioref}
\usetikzlibrary{matrix,positioning}

\usepackage[version=4]{mhchem}

\DeclareMathOperator{\diag}{diag}

\usepackage[table]{xcolor}
\usepackage{chemformula}
\usepackage{enumerate}
\usepackage[all]{xy}
\usepackage{graphicx}
\usepackage{pifont,mathtools}

\usepackage{hyperref}
\usepackage{float}

\usepackage{setspace}
\usepackage{dsfont}
\usepackage{geometry}
\usepackage{enumitem}
\usepackage{todonotes}
\usepackage{bbm}
\usepackage[utf8]{inputenc}
\usepackage[english]{babel}
\usepackage{textcomp}

 \usepackage{algorithm}
\usepackage{algpseudocode}
\usepackage{comment}
\usepackage{tikz-cd}
\usepackage{hyperref}

\usepackage[style=authoryear,  maxcitenames=1,
    mincitenames=1,maxbibnames=99,backend=biber, giveninits=true, dashed=false]{biblatex}
\renewbibmacro{in:}{}
\DeclareNameAlias{author}{last-first}

\DeclareFieldFormat{articletitle}{#1}
\DeclareFieldFormat[journal]{title}{#1}
\DeclareFieldFormat[article]{pages}{#1}
\DeclareFieldFormat{doi}{%
  \addcolon\space%
  \ifhyperref
    {\href{https://doi.org/#1}{\nolinkurl{#1}}}
    {\nolinkurl{#1}}%
}

\renewbibmacro*{journal+issuetitle}{%
  \usebibmacro{journal}%
  \setunit*{\addspace}%
  \printfield{volume}%
  \setunit{\addcolon}
  \newunit%
}

\renewbibmacro*{finentry}{
  \finentry%
}

\makeatletter
\newcommand\level[1]{%
  \ifcase#1\relax\expandafter\chapter\or
    \expandafter\section\or
    \expandafter\subsection\or
    \expandafter\subsubsection\else
    \def\next{\@level{#1}}\expandafter\next
  \fi}
\newcommand{\@level}[1]{%
  \@startsection{level#1}
    {#1}
    {\z@}%
    {-3.25ex\@plus -1ex \@minus -.2ex}%
    {1.5ex \@plus .2ex}%
    {\normalfont\normalsize\bfseries}}

\newcounter{level4}[subsubsection]
\@namedef{thelevel4}{\thesubsubsection.\arabic{level4}}
\@namedef{level4mark}#1{}
\count@=4
\loop\ifnum\count@<100
  \begingroup\edef\x{\endgroup
    \noexpand\newcounter{level\number\numexpr\count@+1\relax}[level\number\count@]
    \noexpand\@namedef{thelevel\number\numexpr\count@+1\relax}{%
      \noexpand\@nameuse{thelevel\number\count@}.\noexpand\arabic{level\number\numexpr\count@+1\relax}}
    \noexpand\@namedef{level\number\numexpr\count@+1\relax mark}####1{}}
  \x
  \advance\count@\@ne
\repeat
\makeatother
\newtheorem{theorem}{Theorem}[section]
\newtheorem{lemma}[theorem]{Lemma}

\newtheorem{prop}[theorem]{Proposition}

\theoremstyle{definition}
\newtheorem{definition}{Definition}[section]

\theoremstyle{remark}
\newtheorem{remark}[theorem]{Remark}

\def\1{\mathbf{1}}

\def\E{\mathbb{E}}
\def\R{\mathbb{R}}

\def\Z{\mathbb{Z}}

\numberwithin{equation}{section}

\begin{document}

\title[Stability and feasibility of Microbial Consumer-Resource Model]{\vspace*{-1.4cm}Stability and feasibility of microbial Consumer-Resource Model}

\author{
Louis Faul\textsuperscript{1}, 
Xavier Richard\textsuperscript{2}, 
Jan Roelof Van der Meer\textsuperscript{3},
Christian Mazza\textsuperscript{1}}

\thanks{\textsuperscript{1} Department of Mathematics, University of Fribourg, Switzerland}
\thanks{\textsuperscript{2} University of Applied Sciences of Western
Switzerland (HES-SO), Switzerland}
\thanks{\textsuperscript{3} Department of Fundamental Microbiology, University of Lausanne, Switzerland}

\thanks{\textbf{Corresponding author:} Christian Mazza, Email: \href{mailto:christian.mazza@unifr.ch}{christian.mazza@unifr.ch}}

\date{}
\begin{abstract}
Microbial communities are ubiquitous in nature but how they grow on available resources is still poorly understood. Communities are complex systems harboring thousands of microbial species that interact through resource competition. The classical MacArthur consumer-resource model has been shown to underestimate formed community biomass. A recent new microbial consumer-resource model (MiCRM) considers the inclusion of inter-specific interactions mediated by metabolite exchange (cross-feeding)  where the various bacterial growth byproducts can be reused by other species for their own growth. We study persistence, feasibility and stability for MiCRM under some simplifying assumptions using slow-fast approximation. We show e.g. the non-persistence of the microbial community when the number of resource species $M$
is smaller than the number of consumer species $S$. We then study the stability of the slow steady state when the number or survivors $S^*$ is smaller than $M$, and show that such equilibria are generically stable. We finally propose a stochastic slow-fast version of the model having fast Poisson steady state and study related extinction events.
 \end{abstract}

\clearpage\maketitle
\thispagestyle{empty}
\maketitle
\vspace*{-0.6cm}

\section{Introduction}\label{Introduction}

Microbial communities have a critical role in nature, and understanding their dynamical properties is essential as they contribute to human, plant, animal and environmental health \parencite{van2022ecological}. 
 Mathematical models of growth for single microbial species in isolation lead to good predictions, the most famous of them being the Monod model \parencite{monod}. Natural microbial communities can harbor thousands of species that interact indirectly through resource competition.  MacArthur \parencite{MacArthur69,MacArthur70} was the first to propose a model taking into account both species and resource dynamics, leading to his famous Consumer-Resource Model (CRM). 
 MacArthur studied a simplified version of the CRM by focusing on a
slow-fast approximation, assuming that resources attain equilibrium much faster than species. 
In this proxy and for fixed species abundances, one solves steady-state equations for resources, and the resulting species dynamics
reduces to classical Lotka-Volterra competition model, whose stability is ensured by a Lyapunov function. \\


Natural microbial ecosystems are driven by inter-specific interactions mediated by metabolite exchanges, where diverse bacterial metabolic byproducts can be reused  by other species for their own growth, corresponding to a cross-feeding mechanism, see, e.g.,
 \parencite{culp2023cross}. Previous studies have shown that models neglecting metabolite-driven interspecific interactions underestimate community biomass \parencite{Guex} and that generalized Lotka-Volterra models cannot reproduce
observed microbial growth patterns, see, e.g. \parencite{momeni,mustri}.
Marsland et al. \parencite{marsland19}
introduced the so-called Microbial Consumer-Resource Model (MiCRM), which explicitly incorporates cross-feeding.
In this model only a fraction of the energy received by a species will be used for its own growth, while the rest will be transformed into metabolites potentially useful for other species.\\

Interesting properties of the MiCRM have been established using methods from statistical physics \parencite{mehta2021cross}.
In this work, we  study the stability of a slow-fast approximation of the MiCRM, and then turn to a stochastic version of this approximation using Markov chains. The paper is organized as follows. Section 2 recalls the basic facts of the MacArthur model and introduces the MiCRM for $M$ resources or metabolites and $S$ consumer species (i.e. the microbial taxa). Section 3 analyzes resource dynamics under the slow-fast approximation, proving equilibrium stability. We provide a closed form 
formula for the fast resource equilibrium under uniform cross-feeding,
 where the energy excess due to metabolite consumption is redistributed uniformly among possible metabolites. Section 4 studies the feasibility, persistence and stability of the slow species o.d.e..
 Such questions have been studied previously for Lotka-Volterra predator-prey systems using random matrix theory, see, e.g. \parencite{stone,dougoud,clegg,mehta2021cross,jamal1,jamal2}, indicating that mass extinctions occur generically  as $S\to\infty$ when the parameters defining the system are chosen at random.
 Concerning  the MiCRM,
we show  the  almost sure non-persistence for the slow consumer o.d.e. when
 $M<S$ and
 the interaction parameters of the model are chosen at random with continuous densities.
  This suggests mass extinctions of consumers that lead to a number of surviving consumer species $S^*$ with $S^*<M$, which is a proxy of the competitive exclusion principle. 
 Using tools from statistical mechanics, the authors of \parencite{mehta2021cross} suggested a species packing upper bound given by $S^*/M\le 1/2$.
 We then consider the stability of  the slow consumer species o.d.e. for steady states with $S^*<M$ and show that such equilibria are generically stable. Exploiting the fact that the fast resource o.d.e. is the mass-action kinetics of a first order chemical reaction network (CRN) of zero deficiency, we propose in Section 5 a stochastic version of the MiCRM having fast Poisson steady state. We then study the stochastic properties of a plausible slow-fast time continuous time Markov chain associated  to the MiCRM.

\section{Preliminaries on Consumer-Resource models \label{Preliminaries}
}
\subsection{Notations}
For any integers $m$ and $n$, $\mathbb{Z}_{\geq 0}^m$ represents the m-dimensional lattice with non-negative integer, $\mathbb{R}^m_{\geq 0}$ is the subset of $\mathbb{R}^m$ whose coordinates are all non-negative and $\mathbb{R}^m_{> 0}$ is the subset of $\mathbb{R}^m$ whose coordinates are all positive.
The set of $n \times m$ matrices with real entries is written $\mathbb{R}^{n\times m}$. The transpose of a matrix $A$ or a vector $x$ is denoted $A^T$ and $x^T$. Let $\vec{0} = (0,...,0)^T \in \mathbb{R}^n$ correspond to the vector with only zero components, while $\vec{1} = (1,...,1)^T \in \mathbb{R}^n$ refers to the vector with only ones. For a column vector $x = (x_1,...,x_d)^T \in \mathbb{R}^d$, or equivalently for a line vector $x = (x_1,...,x_d) \in \mathbb{R}^{1 \times d}$, we denote by $\diag(x_\alpha)_{\alpha \in \{1,...,d\}}$, the diagonal matrix of $\mathbb{R}^{d\times d}$ with diagonal terms $x_1,...,x_d$. We set $(x)_i  \equiv x_i$ for any $i \in \{1,..d\}$. Lastly for two integers $i$ and $j$, $\delta_{i,j}$ is the dirac function with value $1$ if $i=j$ and $0$ otherwise.

\subsection{MacArthur consumer-resource model (MCRM)\label{subsec: MacArthur}}

MacArthur developed seminal consumer-resource models which are at the heart of our investigations 
\parencite{MacArthur69,MacArthur70}.  Let $R_\alpha$, $\alpha=1,\ldots, M$ be the abundances of byproducts or waste species and let $y_i$,
$i=1,\ldots, S$, be the abundances of the consumers. The following set of consumer-resource ordinary differential equations (o.d.e.) gives the rate of growth of
species and resources
\begin{equation*}\label{eq: MacArthurBiomass}
\frac{dy_i}{dt}=g_i y_i \Big(\sum_{\alpha=1}^M c_{i\alpha}w_\alpha R_\alpha -m_i\Big)
\end{equation*}
\begin{equation*}\label{eq: MacArthurResource}
\frac{dR_\alpha}{dt}= \frac{r_\alpha}{K_\alpha}(K_\alpha-R_\alpha)R_\alpha -\sum_{i=1}^S c_{i\alpha}y_iR_\alpha                               
\end{equation*}
where $w_\alpha$ is the weight of item of resource $R_\alpha$ in gram, $g_i$ is the conversion factor from energy uptake to growth rate,
$c_{i\alpha}$ is the rate at which an individual of species $i$ of concentration $y_i$ encounters and uptakes an item
of resource $R_\alpha$ per unit of time, $K_\alpha$ is is the carrying capacity of the habitat for resource $R_\alpha$, $r_\alpha$ is the intrinsic rate of natural increase and $m_i$ is the mortality rate of species $i$.

\subsubsection{Slow-fast approximation\label{subsubsec: McArthurSlowFast}}

To reduce the dimensionality of the system, MacArthur \parencite{MacArthur69} assumed that the resources go much faster to extinction than the species, i.e. he assumed a quasi-equilibrium by setting $dR_\alpha/dt =0$, so that
\begin{equation*}\label{McArthurQuasi}
R_\alpha =K_\alpha-\sum_j c_{j\alpha} \frac{K_\alpha}{r_\alpha}y_j,
\end{equation*}
which is then plugged into the consumer equation to arrive at the competition equation
\begin{equation}\label{eq: Competition}
\frac{dy_i}{dt}=\Big({\mathcal K}_i-\sum_{j=1}^n a_{ij}y_j 
\Big)y_i,
\end{equation}
with
$${\mathcal K}_i = \sum_\alpha c_{i\alpha}w_\alpha K_\alpha -m_i \hbox{  and  }a_{ij} =\sum_\alpha c_{i\alpha}c_{j\alpha} \frac{w_\alpha K_\alpha}{r_\alpha}.$$
\begin{figure}
    \centering
    \includegraphics[width=0.7\linewidth]{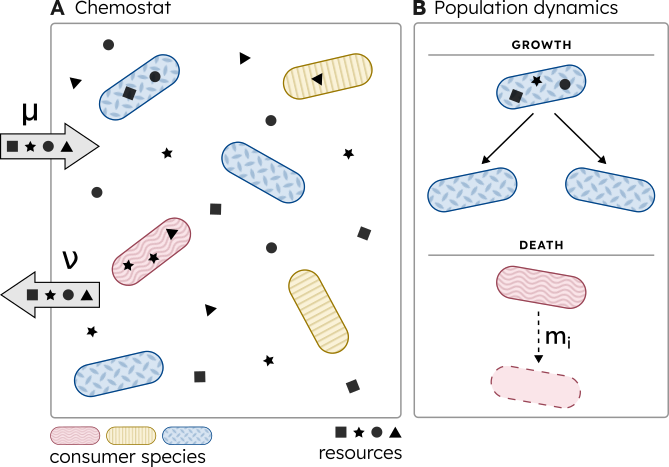}
    \caption{MacArthur model with externally supplied resources. Left panel shows the chemostat constituted of consumers species and resources. Right panel depicts growth and death of a consumer species.}
    \label{fig:sanscross-feeding}
\end{figure}

Equation (\ref{eq: Competition}) is of Lotka-Volterra type and has been widely studied in the literature. In particular, notice that, within this quasi-equilibrium where resources reach equilibrium much faster than consumers, the competition matrix is symmetric. MacArthur used this symmetry to get a Lyapunov function that shows that the orbit of  the o.d.e. reaches the equilibrium point.

When the resources are externally supplied as in Figure \ref{fig:sanscross-feeding}, MacArthur model becomes 
\begin{equation*}
\frac{dy_i}{dt}=y_i \Big(\sum_{\alpha=1}^M c_{i\alpha} R_\alpha -m_i\Big)
\end{equation*}
\begin{equation*}
\frac{dR_\alpha}{dt}= \mu_\alpha - \nu_\alpha R_\alpha -\sum_{i=1}^S c_{i\alpha}y_iR_\alpha,                               
\end{equation*}
for production and degradation rates $\mu_\alpha$ and $\nu_\alpha$.
The next Section \ref{s.MiCRM} generalizes this model by considering metabolite exchange through cross-feeding mechanism.

\subsection{The Microbial consumer-resource model (MiCRM) }\label{s.MiCRM}
Metabolic exchanges  are known to play a major role in microbial ecosystems, but are not considered in the MacArthur model.
Marsland et al. \parencite{marsland19} introduced cross-feeding, so that species not only compete through resources consumption, but also cooperate by exchanging metabolites. The resulting model was called the Microbial consumer-resource Model (MiCRM), and has been applied sucessfully to fit to experimental data, see \parencite{marsland19}, or \parencite{Guex} for an equivalent model based on chemical reaction networks.
The MiCRM model of
 \parencite{mehta2021cross,Marsland2019}
assumes that a fraction $0\leq l \leq 1$ of the energy is released back into the environment as metabolic byproducts. The remaining fraction is used as usual for species growth, see e.g. Figure \ref{fig:cross-feeding}. The model is given by the o.d.e. system
\begin{equation}\label{dynamic_species}
\frac{dy_i}{dt} = g_iG_i(R,y)=g_iy_i \Big(\sum_{\alpha} (1-l_{\alpha})c_{i\alpha} R_{\alpha} -m_i\Big),
\end{equation}
\begin{equation}\label{dynamic_resources}
        \frac{dR_{\alpha}}{dt} = F_\alpha(R,y)=\mu_{\alpha} - \nu_{\alpha} R_{\alpha} - R_{\alpha}\sum_{i}  c_{i \alpha} y_i+ \sum_{i \beta} l_{\beta} D^i_{\alpha\beta} c_{i \beta} R_{\beta} y_i
\end{equation}

where $(D^i)^T $ is the stochastic cross-feeding matrix associated with species $i$, and where  $\mu_\alpha$ and $\nu_\alpha$ are the production and degradation rates associated to  resource $\alpha$. The matrix entry
 $D^i_{\alpha\beta}$ encodes the fraction of resource $\beta$ leaked back to the environment as  resource $\alpha$ when consumed by species $i$. 

  \begin{figure}[h]
      \centering
    \includegraphics[width=0.7\linewidth]{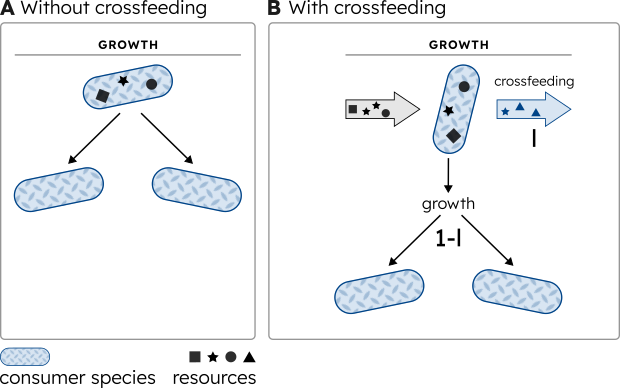}
     \caption{Utilization of resources by consumer species. Left panel shows that, in a model without cross-feeding, all the energy is used for growth. Right panel depicts the situation in the presence of cross-feeding. Only a fraction $1-l$ of the energy is used for growth while the rest is released back into the environment as metabolic byproducts.}
     \label{fig:cross-feeding}
 \end{figure}

Similarly as what is done in \parencite{marsland2020minimal}, we assume that due to the universality of metabolism, $D^i_{\alpha\beta}$ is the same for all species $i$, that is, we set $D^i \equiv D$, for all $i \in \{1,...,S\}$.

Some particular cases are of interest:
when $l = 0$, all the energy
for the microbes is used for growth, and the MiCRM model reduces to the
McArthur model. At the other extreme $l = 1$, all the energy is used for cross-feeding. 

\subsubsection*{Random interaction and metabolite exchange matrices}

Most natural microbial communities exhibit high microbial diversity with a very large number $S$ of  species.
Experimental studies, however, mostly focus on low diversity microbiota.
Even for such rather limited population diversity, the knowledge on interaction and
metabolite exchange coefficient $c_{i\alpha}$ and $D_{\alpha\beta}$
is very scarse. To overcome this fundamental issue, and to get information on
structural properties of the MiCRM model, 
most studies suppose that
 both the interaction matrix
 $c \in \mathbb{R}^{S \times M}$ and the metabolite exchange matrix $D\in \mathbb{R}^{M \times M}$ are randomly generated,  
 see \parencite{Marsland2019,marsland2020minimal,marsland2020community,mehta2021cross}.
 In these studies, the entries of the random matrix $c_{i\alpha}$ are drawn independently from a distribution with scaled mean and variance, for which $c_{i\alpha} >0$. A usual choice is a gaussian distribution truncated in $0$.
The cross-feeding matrix $D$ has to be drawn such that it is a column-stochastic matrix. In these  studies, the set of metabolites is divided into several classes,  where intra-classes and inter-classes cross-feeding interactions vary among classes. A Dirichlet distribution with shape parameters scaling as $M$, is then well adapted for sampling $D$.\\

\subsubsection*{Random interaction and uniform cross-feeding}

Some of the results of this work assume a non-negative random interaction matrix $c$ with i.i.d. entries and random i.i.d. non-negative random variables $m_i$.
Structural results that are obtained in such random frameworks are clearly indicated.
We, however, mostly assume uniform cross-feeding by setting
$D = \frac{1}{M} \vec{1} \vec{1}^T$, where $\vec{1}$ is the column vector composed of ones.

We proceed as in the McArthur's model given in Section \ref{subsubsec: McArthurSlowFast}
where slow-fast approximation is considered, with fast ressource and slow species dynamics.

\section{Resource dynamics under slow-fast approximation\label{ConsumerFast}}

We will see in the following Section \ref{MatrixFormulation}, that the o.d.e. system (\ref{dynamic_species},\ref{dynamic_resources}) is of the form
$$
    \frac{dR}{dt} = F(R,y),\ \ 
    \frac{dy}{dt} = g G(R,y),$$ 
where $F(R,y) = \mu - RB_l(y)$ and $G(R,y)= \diag(y_i)_{i \in \{1,...,S\}} \left ( (1-l)C^T R - m \right)$.

When the efficiency proportionality constant $g_i$ is often assumed to be very small i.e. $0 < g \ll 1$, see e.g. \parencite{chesson1990macarthur,poggiale2020analysis}, the system exhibits two different timescales, and is called a \textbf{slow-fast system}, see Section \ref{subsubsec: McArthurSlowFast}.
The mathematical properties of such systems are given by a Theorem of Tykhonov, see, e.g., \parencite{lobry}. 
The slow-fast approximation is given as a pair of o.d.e.:
\begin{itemize}
    \item The degenerate or slow system:
    \begin{equation}\label{degenerate}
    \frac{dy}{dt} = g G(R^*(y),y),
    \end{equation}
where the equilibrium $R^*(y)$ is the solution of the equation $F(R^*(y),y)=0$, 
\item and the fast system 
\begin{equation*}\label{fastsystem}
    \frac{dR}{dt} = F(R,y),
\end{equation*}
where $y$ is treated as a constant.
\end{itemize}
We  show in Section \ref{stability_fast} that the equilibrium $R^*(y)$ is a stable solution of the degenerate system. Tikhonov's Theorem then ensures that the solution $y(t)$ of the system (\ref{dynamic_species},\ref{dynamic_resources}) converges toward the solution of the degenerate system (\ref{degenerate})
as  $g \to 0$.

Hence, as in the MacArthur model, we  consider that the resources will reach equilibrium much faster than the species, i.e. we assume quasi-equilibrium by setting $dR_{\alpha}/dt = 0$. 
In what follows, species abundances  will be fixed to some positive values 
$y_i\in\R_{ > 0}$  for $i \in \{1.\cdots,S\}$.
We first show  that the fast system   has a unique asymptotically stable equilibrium $R^*(y)$, $\forall y \in\R_{ >0}^S$.

\subsection{Matrix reformulation of the fast system\label{MatrixFormulation}}
Let rewrite the system in matrix form in order to use some classic results and simplify the computations.
We consider the line vectors $R=(R_1,\ldots,R_M)$, $\nu=(\nu_1,\ldots,\nu_M)$ and $\mu=(\mu_1,\ldots,\mu_M)$. For $0<l\leq 1$, set
\begin{equation}\label{Laplacienl}
L_l = lP- I_M,
\end{equation}
where $P=D^T$ is the stochastic metabolite exchange matrix.
Note that $-L_1$ is a discrete Laplacian matrix, and $-L_l$ is a discrete sub-Laplacian matrix, for $0<l<1$.  
For any fixed state of the slow system $y=(y_1,...,y_S)^T$ and any resource $\alpha$, the global uptake rate is defined by
$$x_{\alpha}(y) =\sum_i y_i c_{i\alpha} > 0.$$
Let
 $$ \Gamma_l(y) = \diag\left(x_\alpha(y)\right)_{\alpha \in \{1,...,M\}}L_l\,$$
 for  $0<l\leq 1$.
The  evolution of resources abundances is governed by equation (\ref{dynamic_resources}), which can be  rewritten as a function of $\Gamma_l(y)$ 
\begin{equation}\label{Gen2B}
\frac{dR_\alpha}{dt} = \mu_\alpha -\nu_\alpha R_\alpha +(R\Gamma_l(y))_\alpha,
\end{equation}
or equivalently
\begin{equation*}\label{resources_matricial}
    \frac{dR}{dt} = \mu - RB_l(y),
\end{equation*} 
where $B_l(y) = \diag(\nu_\alpha)_{\alpha \in \{1,...,M\}}-\Gamma_l(y) $.
By setting $dR/dt=0$ at quasi equilibrium, we are looking for equilibrium $R^*(y)$ solving
$\mu - R^*(y)B_l(y) = 0$.

\subsection{Stability of the fast system }\label{stability_fast}
An equilibrium point $R^*$ is linearly stable when all eigenvalues of the Jacobian evaluated at this point have negative real parts, and  is linearly unstable otherwise.
In this section, we will assume that all species have the same degradation rate i.e. $\nu_\alpha \equiv \nu$, $\forall \alpha \in {1,...,M}$.
The Jacobian matrix  $J(y) = -B_l(y)$ of this system  does not depend on the equilibrium of the resources $R^*$.  The spectral bound of a matrix $A \in \mathbb{R}^{n\times n}$ is defined by
 $$s\left(A\right) = \{\max{{\rm Re}(\lambda), Ax = \lambda x \text{ for some } x \in \mathbb{C}^n}\},$$
 For the Jacobian matrix, we have  $s\left(J(y)\right) = s\left(\Gamma_l(y)\right) - \nu$.
We use results of \parencite{chen2022two},\parencite{chen2023impact} to show that the slow system has a unique asymptotically stable equilibrium $R^*(y)$, 
$\forall y\in\R_{ >0}^S$, by proving that $s(\Gamma_l)\le 0.$
\begin{prop}
    The equilibrium point for the resources $R^*(y)$ is stable $\forall y \in\R_{ >0}^S$.
\end{prop}
\begin{proof}
Notice that
$\Gamma_l(y) = l  \diag \left(x_\alpha(y)\right)_{\alpha \in \{1,...,M\}}  P-\diag \left(x_\alpha(y)\right)_{\alpha \in \{1,...,M\}} $.
Theorem 6 of \parencite{altenberg} shows that
$$\frac{{\rm d}}{{\rm dl}}s(\Gamma_l(y))\le s(\diag \left(x_\alpha(y)\right)_{\alpha \in \{1,...,M\}}P).$$
Perron-Frobenius yields that 
$s(\diag \left(x_\alpha(y)\right)_{\alpha \in \{1,...,M\}}P)\le \max_\alpha x_\alpha(y)$.
Moreover
$s(\Gamma_0(y))=-\min_\alpha x_\alpha(y)$,
so that
$s(\Gamma_l(y))\le -\min_\alpha x_\alpha(y) + l \max_\alpha x_\alpha(y)$.
$s(\Gamma_l(y))\le 0$ is thus satisfied when
$$l\le \frac{\min_\alpha x_\alpha(y)}{\max_\alpha x_\alpha(y)}\le 1.$$
On the other hand, $\Gamma_1(y) = \diag \left(x_\alpha(y)\right)_{\alpha \in \{1,...,M\}}(P-I_M)$, so that $-\Gamma_1(y)$ is a discrete Laplacian matrix, with
$s(\Gamma_1(y))=0$. Results of \parencite{altenberg} (see also \parencite{chen2022two}, Theorem 2.2) show that $s(\Gamma_l(y))$ is a convex function of $0<l<1$. We thus conclude
that $s(\Gamma_l(y))\le  0$, $\forall 0<l<1$, which in turn yields that $s(J(y))<0$ since $\nu >0$.
\end{proof}


\subsection{Closed form solution for uniform cross-feeding\label{ClosedForm}}

Let $ y \in\R_{ >0}^S$. To find an equilibrium $R^*(y)$, we set $dR/dt=0$ at quasi equilibrium, and since $B_l(y)$ is invertible for $0\leq l <1$, we find:
$$R^*(y) = \mu B_l(y)^{-1}$$

The following result provides a closed-formed formula for the equilibrium $R^*(y)$ under \textit{uniform cross-feeding. }
We will work in this setting for the rest of the section.


\begin{prop}
    For any $ y \in\R_{ >0}^S$ the equilibrium $R^*_\alpha(y)$ of the resources dynamics (\ref{dynamic_resources}) under uniform cross-feeding is given by
\begin{align*}\label{general_equilibrium_resources}
    R^*_\alpha (y) &= \frac{l\sum_\beta \frac{x_\beta(y) \mu_\beta}{x_\beta(y) + \nu_\beta}}{M(x_\alpha(y)+\nu_\alpha)\left(1 - \frac{l}{M}\sum_\beta \frac{x_\beta(y)}{x_\beta(y) + \nu_\beta}\right)}+\frac{\mu_\alpha}{x_\alpha(y) + \nu_\alpha}.
\end{align*} 
When $\mu_\alpha \equiv \mu $ and $\nu_\alpha \equiv \nu$, we have
\begin{equation}\label{equilibrium_resources}
R^*_\alpha (y) = \frac{\mu}{A(y)}\frac{1}{x_\alpha(y) + \nu},
\end{equation}
where $A(y)=1-l +(\nu l/M)\sum_\alpha \frac{1}{x_\alpha(y)+\nu} > 0$.

\end{prop}

\begin{proof}
In the case of uniform cross-feeding the sub-Laplacian matrix describing the transitions between resources is then given by:
 $$ - (\Gamma_l(y))_{\alpha,\alpha} = \frac{M-l}{M} x_\alpha(y),  \quad (- \Gamma_l(y))_{\alpha,\beta} = -l\frac{x_\alpha(y)}{M}, \quad \forall \alpha \neq \beta.$$ 
     $B_l(y)$ can thus be written as a rank-one perturbation of a diagonal matrix 
    $$B_l(y) = \diag(\nu_\alpha)_{\alpha \in \{1,...,M\}} - \Gamma_l(y) = \diag(\nu_\alpha+x_\alpha(y))_{\alpha \in \{1,...,M\}} + u(y)\vec{1}^T,$$
    where $ u(y) = \left(-\frac{l}{M} x_\alpha(y)\right)_{1 \leq \alpha \leq M} \in \mathbb{R}^{M\times 1}$.

We apply Sherman-Morrison formula \parencite{sherman1950adjustment} to find the inverse of $B_l(y)$, which gives:
\begin{equation*}
(B_l(y)^{-1})_{\alpha,\beta} = \frac{lx_\alpha(y)+\delta_{\alpha,\beta} M (x_\alpha(y)+\nu_\alpha)}{M(x_\alpha(y)+\nu_\alpha)(x_\beta(y)+\nu_\beta)\left(1 - \frac{l}{M}\sum_\gamma\frac{x_\gamma(y)}{x_\gamma(y) + \nu}\right)}.
\end{equation*}
The equilibrium vector follows easily
\begin{align*} 
    R^*_\alpha (y) &= (\mu_1, \mu_2, ...,\mu_M)B_l(y)^{-1}  \\ 
    &= \frac{l\sum_\beta \frac{x_\beta(y) \mu_\beta}{x_\beta(y) + \nu}}{M(x_\alpha(y)+\nu_\alpha)\left(1 - \frac{l}{M}\sum_\beta \frac{x_\beta(y)}{x_\beta(y) + \nu_{\beta}}\right)}+\frac{\mu_\alpha}{x_\alpha(y) + \nu_\alpha}
\end{align*}
When $\mu_\alpha \equiv \mu $ and $\nu_\alpha \equiv \nu$, we get:
\begin{equation*}
    R^*_\alpha (y)= \frac{\mu}{A(y)}\frac{1}{x_\alpha(y) + \nu}
\end{equation*}
where $A(y)=1-l +(\nu l/M)\sum_\alpha\frac{1}{x_\alpha(y) +\nu}$.

\end{proof}

\section{Slow  species dynamics for uniform cross-feeding }\label{section_slow}
We follow \parencite{MacArthur69} by looking at the slow species system by plugging the expression of the resources equilibrium into the species equation (see (\ref{degenerate}))
to arrive to the slow o.d.e. 
\begin{equation}\label{slow_dynamic_species}
\frac{dy_i}{dt} = g_iG_i(y)= g_iy_i \Big(\sum_{\alpha} (1-l_{\alpha})c_{i\alpha} R^*_{\alpha}(y) -m_i\Big).
\end{equation}
We are interested in the  feasibility and linear stability of this system. A direct computation shows that
\begin{equation}\label{Integral}
y_i(t)=g_iy_i(0)\exp(\int_0^t \sum_{\alpha} (1-l_{\alpha})c_{i\alpha} R^*_{\alpha}(y(s)){\rm d}s -m_i  t),
\end{equation}
showing that the orbits of the slow species dynamics remain positive at any time $t$ when
$y_i(0) >0$, $\forall i\in \{1,...,S\}$.\\

 We first consider {\bf feasibility}, that is, we first focus on the existence of positive equilibria $y^*>0 $ with
$G_i(y^*)=0$ $\forall i \in \{1,...,S\}$.
Assume that $\mu_\alpha \equiv \mu $, $l_\alpha \equiv l$ and $\nu_\alpha \equiv \nu$ for all $\alpha \in \{1,...,M\}$. Plugging (\ref{equilibrium_resources})
in (\ref{slow_dynamic_species}) when 
$\frac{dy_i}{dt}=0$ one arrives at the linear system
\begin{equation}\label{LinearSystemSlow}
CZ(y^*) = \frac{A(y^*)}{\mu (1-l)} m,
\end{equation}
where 
 $Z(y^*)$ is the $M$-dimensional vector of entries
given by $z_\alpha(y^*)=1/(x_\alpha(y^*) +\nu)$, $\alpha=1,\cdots, M$. 
Notice that we are looking for positive equilibria with $y_i >0$, so that
$x_\alpha(y) =\sum_i c_{i\alpha}y_i >0$. We thus look for positive solutions $Z(y^*)$ of the linear equation (\ref{LinearSystemSlow}). If no positive solution exists, then  the slow dynamical systems does not have a positive equilibrium, suggesting that some  slow species become extinct.
On the contrary, if (\ref{LinearSystemSlow}) has a positive solution $Z^*(y^*) >0$, one must then check the existence of positive solution $y^* >0$ to the linear system
\begin{equation*}
    \frac{1}{Z^*(y^*)}-\nu (1,\cdots,1)^T = C^T y^*,
\end{equation*}

where $\frac{1}{Z^*(y^*)}$ is the vector of entries $1/z^*_\alpha(y^*)$.\\

Assume that both the interaction matrix $C$ and $m$ are chosen independently at random with positive entries $c_{i\alpha}$ and $m_i$ having continuous probability densities.
When $\lambda = M/S <1$, the dimension of the column space of the matrix $C$ is at most $M< S$. Hence, the linear system 
(\ref{LinearSystemSlow}) will have no solution almost surely (a.s.) when the random vector 
$m$ of $\R_{ >0}^S$ has a density. The probability of feasibility vanishes then under these conditions, that is
$$\mathbb{P}(\exists \ y^*>0 \hbox{ with }G(y^*)=0)=0,$$
where the vector field $G$ is given by (\ref{slow_dynamic_species}).

When e.g. $M=S$ we can rewrite (\ref{LinearSystemSlow}) as $C \bar Z=m$, where $\bar Z= \mu(1-l)Z(y^*)/A(y^*)$.
Results of \parencite{stone,dougoud} show that for random positive matrices $C$ of i.i.d. entries uniformly distributed on a positive interval, the probability of feasibility
$\mathbb{P}( \bar Z_\alpha(y^*) >0,\ \forall \alpha)\approx p^M$ for a parameter $p$ with $0<p<1$, and thus converges exponentially fast toward 0 as $M=S\to\infty.$

Results of \parencite{mehta2021cross} suggest that the number of surviving species $S^*$ should satisfy $S^* \leq M/2 $ in the large-system limit. The MiCRM is thus unlikely to admit a feasible equilibrium when $M/2 < S < M$. 
\subsection{Non-persistence of the slow species  dynamics when $\lambda =M/S<1$\label{subsection_feasibility}}

We focus on persistence of the slow dynamics (\ref{slow_dynamic_species}). The flow induced by the vector field $F_i(y)=g_i G_i(y)$, $i=1,\ldots,S$ is the mapping $\Phi:\ \R\ {\rm x}\ \R^S\to\R^S$, $(t,x)\mapsto \Phi_t(x)$ such that
$t\mapsto \Phi_t(x)$ is the solution of (\ref{slow_dynamic_species}) with initial condition $x\in\R^S$. The positive  trajectory is the set  $\gamma^+(x) = \{\Phi_t(x);\ t\ge 0\}$, and the {\bf omega limit set} of $x$, denoted by $\omega(x)$ is the set of points $p$ of $\R^S$ such that $p=\lim_{n\to\infty}\Phi_{t_n}(x)$, for some sequence $t_n\to\infty$.
\begin{definition}\label{persistent}
  The system is persistent if all the species coexist, that is, if for every bounded trajectory $\gamma^+(x)$
  $$x\in\R_{ >0}^S \Rightarrow \omega(x)\subset\R_{ >0}^S.  $$
\end{definition}

The previous results show that
$$m\not\in K(\R_{ >0}^S),\ a.s.,$$
where the vector field $K$ is given by 
$$K_i(y)=
\sum_{\alpha} (1-l_{\alpha})c_{i\alpha} R^*_{\alpha}(y),$$
$i=1,\cdots,S$ and $m=(m_1,\cdots,m_S)^T\in \R_{ >0}^S$.
We next use the fact that $K(\R_{ >0}^S)$ is  convex to show that the orbits 
of (\ref{slow_dynamic_species}) issued from $y(0)\in \R_{ >0}^S$
leave ultimately $\R_{ >0}^S$ suggesting that one can focus on equilibria
$y^*$ such that $y_j^*\ge 0$ for all $j \in \{1,..,S\}$ with 
$y_i^* =0$  for at least one $i\in\{1,\cdots,S\}$ (see Definition \ref{persistent} and Proposition \ref{nopersistence}).
\begin{lemma}\label{Convexity}
Assume that $\mu_\alpha \equiv \mu $ and $\nu_\alpha \equiv \nu$.  Then
$K(\R_{ >0}^S)$ is  convex.
\end{lemma}
\begin{proof}
We use the fact that convexity is preserved under linear-fractional transformations 
$f:\ \R^n \to \R^m$ of the form $f(x)= A x +b /(c^T x +d)$, for any $m\ {\rm x}\ n$ matrix $A$, $b\in\R^m$, $c\in\R^n$ and $d\in\R$.
We have seen in (\ref{equilibrium_resources}) that
$$R^*_\alpha (y) = \frac{1}{1-l +\frac{\nu l}{M} \sum_\beta \frac{1}{x_\beta(y) + \nu}} \frac{\mu}{x_\alpha(y) + \nu}.$$
Then $R^*_\alpha(y)=f_\alpha(Z(y))$ where
$$f_\alpha(z)=\frac{\mu z_\alpha}{1-l +\frac{\nu l}{M}\sum_\beta z_\beta},$$
and
$$Z_\alpha(y) = \frac{1}{c_\alpha^T y  +\nu},$$
where $c_\alpha$ is the corresponding column vector of the matrix $C$. Hence 
$Z(\R_{ >0}^S)$, $f_\alpha(Z(\R_{ >0}^S))$, $f(Z(\R_{ >0}^S))$ and $K(Z(\R_{ >0}^S))$ are convex.

\end{proof}

\begin{prop}\label{nopersistence}
Assume that $\mu_\alpha \equiv \mu $ and $\nu_\alpha \equiv \nu$. Suppose that $\lambda < 1$. 
Suppose that both the interaction matrix $C$ and $m$ are chosen independently at random with positive entries $c_{i\alpha}$ and $m_i$ having continuous probability densities.
Then the slow species dynamics (\ref{slow_dynamic_species}) is a.s. not persistent.
\end{prop}
\begin{proof}
We follow the proof of Proposition 10.8.3 of \parencite{MazzaBenaim}.
  As we have seen previously, $\lambda < 1$  implies that $m\not\in K(\R_{ >0}^S)$ a.s.. Using the convexity of $K(\R_{ >0}^S)$
  we get that there is an hyperplane through $m$ which does not meet $K(\R_{ >0}^S)$.
  There exists a unit vector $u$ orthogonal to this hyperplane such that
  $(g-m)^T u >0$ for all $g\in K(\R_{ >0}^S)$. Let $V:\ \R_{ >0}^S\to\R$ be the map given by
  $$V(y)= -u^T \ln(y).$$
  Let $y(t)$ be a solution of (\ref{slow_dynamic_species}) issued from $\R_{ >0}^S$. Then
  $$\frac{{\rm d}}{{\rm d}t}V(y(t))=- K(y)^T u <0,$$
  showing that $V$ is a strict Lyapunov function.
  The result follows from the fact that no compact limit set can be contained in $\R_{ >0}^S$,
  see, e.g.,
  Theorem 10.5.1 of \parencite{MazzaBenaim} .
\end{proof}

\subsection*{Competitive exclusion}
Proposition \ref{nopersistence} shows that under some conditions, the slow species dynamics is not persistent, suggesting that the only allowed equilibria $y^*$
 are those for which $y_i^* =0$ for some species $i$. 
 Let $\mathcal{I}=\{1\le i\le S;\ y_i^* > 0\}$ be the set of surviving species with $\vert \mathcal{I}\vert = S^*$.
 The authors of \parencite{mehta2021cross} argued, using methods of physics and numerical simulations, that
 $S^* \leq M/2$. This is stronger than the well-known {\it competitive exclusion principle} which suggests that
 one should have $S^* \le M$. The following simple argument suggests that $S^* \le M$ is necessary to ensure the coexistence of surviving species. Assume that $y_i(t) \to y_i^* >0$ for $i\in \mathcal{I}$ and 
 $y_i(t)\to 0$ for $i\not\in \mathcal{I}$. Under reasonable assumptions, 
 $$x_\alpha(t) \underset{t \to \infty}{\to} \sum_{i\in \mathcal{I}}c_{i\alpha} y_i^*,$$
 $$\frac{1}{t}\int_0^t R_\alpha^*(y(s)){\rm d}s \underset{t \to \infty}{\to} R_\alpha^*(y^*).$$
The basic formula (\ref{Integral}) then gives that 
$R^*$ belongs to the following hyperplane  $H_i^0$ of $\R^M$ when $i\in \mathcal{I}$
\begin{equation*}
H_i^0:\ \sum_\alpha (1-l_\alpha)c_{i\alpha }R_\alpha-m_i = 0,
\end{equation*}
and that it belongs to the half-space $H_i^-$ for $i\not\in\mathcal{I}$
\begin{equation}\label{Cond1}
H_i^-:\ \sum_\alpha (1-l_\alpha)c_{i\alpha }R_\alpha-m_i < 0.
\end{equation}
Hence this suggests that
$$R^*(y^*) \in \mathop{\bigcap}_{i\in \mathcal{I}}H_i^0\cap \mathop{\bigcap}_{i\in \mathcal{I}^c}H_i^-.$$
Assuming that the set of hyperplanes $H_i^0$, $i\in \mathcal{I}$ is in general position, 
$${\rm dim}(\mathop{\cap}_{i\in \mathcal{I}}H_i^0)= M-S^*,$$
showing that $S^*\le M$ is required to get an non-empty intersection.

\subsection{Boundedness of slow species trajectories}

\begin{lemma}
    Every birth rate $b_i(y)=g_i(1-l)\sum_{\alpha}c_{i,\alpha} y_i  R_{\alpha}^*(y)$ is bounded by $ \sum_\alpha \mu_\alpha$.
\end{lemma}
\begin{proof}
    Consider equilibrium in (\ref{dynamic_resources}) :
\begin{align*}
    0 = \mu_{\alpha} - \nu_{\alpha} R_{\alpha}^*(y) -  \sum_{i}  c_{i,\alpha} y_i R_{\alpha}^*(y)  + l  \sum_{i,\beta} c_{i,\beta} y_i  R_{\beta}^*(y) P_{\beta \alpha},
\end{align*}
which leads to
$$ 0 =\sum_{\alpha} \mu_{\alpha} - \sum_{\alpha} \nu_{\alpha} R_{\alpha}^*(y) -  \sum_{i,\alpha}c_{i,\alpha}  y_i R_{\alpha}^*(y) + l  \sum_{i,\beta} y_i c_{i,\beta} R_{\beta}^*(y) \underbrace{\sum_{\alpha} P_{\beta \alpha} }_{=1}.$$ 
The positivity of the various coefficient $c_{i\alpha}$ and of $R^*_\alpha(y)$ 
shows that
\begin{equation}\label{bound_birth}
    (1-l)\sum_{i,\alpha}c_{i,\alpha} g_i y_i  R_{\alpha}^*(y) = g_i\left(\sum_{\alpha}\mu_{\alpha} - \sum_{\alpha} \nu_{\alpha} R_{\alpha}^*(y)\right) \leq g_i\sum_{\alpha}\mu_{\alpha} \leq \sum_{\alpha}\mu_{\alpha}.
\end{equation}
The sum of the birth rates $\sum_i b_i(y)=(1-l)\sum_{i,\alpha}c_{i,\alpha}g_i y_i  R_{\alpha}^*(y) $ is thus bounded by $\sum_\alpha \mu_\alpha$ so that any of its positive term is also bounded by the same constant.
\end{proof}

\begin{prop}
    All species abundances are bounded over time.
\end{prop}
\begin{proof}
Recalling that we are dealing with the solutions of the slow equation
\begin{equation*}
\frac{dy_i}{dt}=g_iy_i \Big(\sum_{\alpha} (1-l)c_{i\alpha} R^*_{\alpha} (y)-m_i\Big),
\end{equation*}
and using the upper bound (\ref{bound_birth}), we arrive at
$$0\le y_i \sum_{\alpha} (1-l)c_{i\alpha} R^*_{\alpha}(y) \leq \sum_\alpha \mu_\alpha.$$
We get therefore the following differential inequalities 
\begin{equation*}
    \frac{dy_i}{dt} \leq \left( \sum_\alpha \mu_\alpha - y_i m_i\right).
\end{equation*}
This a special case of the Grönwall inequality under its differential form \parencite{mischler}. We get that :
\begin{equation*}
    y_i(t) \leq y_i(0)e^{-m_i t}-\frac{\sum_\alpha \mu_\alpha}{m_i}(e^{-m_i t}-1)
\end{equation*}
This bound ensures that there will be no divergence of the species abundances to infinity, but is of course not optimal.
\end{proof}

\subsection{Stability for slow species system}
 Let $y^*$ be an equilibrium of the slow o.d.e. system. The results of Section \ref{subsection_feasibility} show that
 one should focus on the case where the number of surviving species $S^* = \vert \mathcal{I}\vert$ satisfies $S^*\le M$.
 The Jacobian $J(y^*)$ at some equilibrium  $y^*$ is given by:
\begin{align*}\label{expression_jacobian}
    J_{i,j}(y^*) = \frac{\partial G_i(y^*)}{\partial y_j} = g_i\delta_{i,j} \left[ (1-l) \sum\limits_\alpha c_{i\alpha} R^{*}_\alpha (y^*) -m_i \right] + g_i y_i^* (1-l)\sum\limits_\alpha c_{i\alpha} \frac{\partial  R_\alpha^*(y^*)}{\partial y_j}
\end{align*}

For a given realization of the system, one can always re-arrange the species such that the $S^*  $ first species survives and the $S-S^*$ other species go extinct. Following \parencite{baron2023breakdown}, the jacobian can be written as the following block matrix:
$$J(y^*) = \begin{pmatrix}
    J'(y^*) & \Pi(y^*) \\
    0 & \Omega(y^*)
\end{pmatrix}$$
where $J'(y^*) \in \mathbb{R}^{S^* \times S^*}$ is the reduced jacobian matrix i.e. the jacobian restricted to surviving species. The eigenvalues of $J(y^*)$ are the combination  of those of $J'(y^*)$ and those of $\Omega(y^*)$.
The ones of $\Omega(y^*)$ are its diagonal terms and are negative according to (\ref{Cond1}), and thus linear stability of the equilibrium $y^*$ only depends on the eigenvalues of $J'(y^*)$. The upper right block $\Pi(y^*)$ does not influence the eigenvalues of $J(y^*)$.
\begin{prop}\label{main_prop}
Assume uniform cross-feeding and
 consider constant creation and degradation rates with $\mu_\alpha \equiv \mu$ and $\nu_\alpha \equiv \nu$, for all $\alpha \in \{1,...,M\}$. 
    Then the equilibria of the slow-system are linearly stable.
\end{prop}

\begin{proof}
Let $y^*$ be an equilibrium of the slow o.d.e. system. In this proof, we shall for clarity omit the dependency on $y^*$ and denote $A=A(y^*)$ and $x_\alpha = x_\alpha(y^*)$.  For any $j \in \{1,...,S^*\}$:
\begin{equation*}
    \frac{\partial R_\alpha(y^*)}{\partial y_j} = \mu \left( \frac{\frac{\nu l}{M}\sum_\beta \frac{c_{j\beta}}{(x_{\beta}+\nu)^2}}{(1-l+\frac{\nu l }{M}\sum_\beta \frac{1}{x_\beta+ \nu})^2}\frac{1}{x_\alpha + \nu}-\frac{c_{j\alpha}}{(1-l+\frac{\nu l }{M}\sum_\beta \frac{1}{x_\beta + \nu})(x_\alpha+\nu)^2}\right)
\end{equation*}  and so for any $i,j \in \{1,...,S^*\}$:
\begin{equation}\label{Jacobian}
    J'_{ij}(y^*) = g_iy^*_i (1-l) \mu \left( \frac{\frac{\nu l}{M}\sum_\beta \frac{c_{j\beta}}{(x_{\beta}+\nu)^2}}{(1-l+\frac{\nu l }{M}\sum_\beta \frac{1}{x_\beta + \nu})^2}\sum_\alpha \frac{c_{i\alpha}}{x_\alpha + \nu}- \frac{1}{1-l+\frac{\nu l }{M}\sum_\beta \frac{1}{x_\beta + \nu}}\sum_\alpha \frac{c_{i\alpha} c_{j\alpha}}{(x_\alpha+\nu)^2}\right).
\end{equation} We thus get the following formula
\begin{equation*}
    J'(y^*) = -\mu (1-l) \diag(g_iy_i)_{i \in \{1,...,S\}}^* \left(\frac{\frac{\nu l}{M}}{(1-l+\frac{\nu l }{M}\sum_\beta \frac{1}{x_\beta + \nu})^2} uv^T -\frac{1}{1-l+\frac{\nu l }{M}\sum_\beta \frac{1}{x_\beta(y^*) + \nu}}\tilde C \tilde{C}^T \right)
\end{equation*} 
where $u = (
    \sum_\beta\frac{c_{1\beta}}{x_\beta+\nu},\ldots, 
    \sum_\beta\frac{c_{S^*\beta}}{x_\beta+\nu})^T$,
 $v = 
    (\sum_\beta\frac{c_{1\beta}}{(x_\beta+\nu)^2},\ldots,
    \sum_\beta\frac{c_{S^*\beta}}{(x_\beta+\nu)^2})^T$,
 and $\Tilde{C} = C^* \cdot \diag\left(\frac{1}{x_\alpha+\nu}\right)_{\alpha \in \{1,...,M\}} $
 and where $C^*$ is the submatrix of $C$ associated to the lines with $i\in  \mathcal{I}$.
 Notice furthermore that
 $$\sum_\alpha \frac{c_{i\alpha}}{x_\alpha+\nu}
 = \frac{A}{\mu} \sum_\alpha c_{i\alpha}R_\alpha^*(y^*),
 $$
 where we use (\ref{equilibrium_resources}) for
  an equilibrium $y=y^*$  of the slow system and $i\in \mathcal{I}$ which is such that
 $$(1-l)\sum_\alpha c_{i\alpha} R_\alpha^*(y^*)=m_i.$$
 Hence,
 \begin{equation*}\label{form}
 \sum_\alpha \frac{c_{i\alpha}}{x_\alpha +\nu}=\frac{A m_i}{\mu (1-l)}. 
 \end{equation*}
 (\ref{Jacobian}) then becomes
 $$
     J'_{ij}(y^*) =\frac{y^*_i \mu (1-l)}{A} 
 \Big(
 \frac{\nu l m_i}{\mu (1-l)M}\sum_\beta \frac{c_{j\beta}}{(x_\beta +\nu)^2}
 -\sum_\alpha \frac{c_{i\alpha}c_{j\alpha}}{(x_\beta +\nu)^2}\Big).$$
 The reduced Jacobian is thus the sum of two matrices, the first one being of rank one and the second being proportional to 
 $-
 \tilde C \tilde{C}^T$. 

We first show that the real parts of the eigenvalues of the matrix $\Lambda = (\Lambda_{ij})_{1 \leq i,j \leq M}$ where
$$\Lambda_{ij} = 
 \frac{\nu l m_i}{\mu (1-l)M}\sum_\beta \frac{c_{j\beta}}{(x_\beta(y^*) +\nu)^2 }
 -\sum_\alpha \frac{c_{i\alpha}c_{j\alpha}}{(x_\beta(y^*)+\nu)^2},
$$
are negative.
Let $c_\alpha\in\R^{S^*}$ be the column vector of $C$ associated  to $\alpha$, and set
$w_\alpha = z_\alpha^2$ with $z_\alpha = 1/(x_\alpha+\nu)$. Then
$$\Lambda = 
C \sum_\alpha w_\alpha\big(\gamma z e_\alpha^T -e_\alpha e_\alpha^T\big)C^T,$$
where $\gamma = \nu l/(AM)$,
 $e_\alpha $ is the canonical unit vector of $\R^M$ corresponding to $\alpha$, and 
$z=(z_\alpha)_{\alpha =1,\cdots,M}$. $\Lambda$ is  non-positive definite if and only if
the reduced matrix $\gamma z w^T-\diag(w_\alpha)_{\alpha \in \{1,...,M\}}$ is non-positive definite, or, equivalently, if
$$U= \gamma w z^T-\diag(w_\alpha)_{\alpha \in \{1,...,M\}},$$
where $w=(w_\alpha)_{\alpha =1,\cdots,M}$ is non-positive definite.
The edges of the spectrum of its symmetrized version 
$S=(U+U^T)/2$ provide lower and upper bounds for the real parts of the eigenvalues of $U$ (Bendixson Theorem, see e.g. \parencite{stoer1980introduction}). It is then sufficient to show that
 $x^T S x \le 0$ $\forall x\in \R^M$. In fact
$x^T S x = x^T U x = x^T U^T x$, so that we can equivalently  show that
$x^T U x \le 0$ $\forall x\in \R^M$.
Consider the matrix $Q=(q_{\alpha\beta})$ given by
$$q_{\alpha\beta}= \gamma w_\alpha z_\beta,\ \beta\ne\alpha,$$
$$q_{\alpha\alpha}= w_\alpha (d-1)+\gamma w_\alpha z_\alpha,$$
where $d=1-\gamma\sum_\alpha z_\alpha$ is such that
the row sums of $Q = d\ \diag(w_\alpha)_{\alpha \in \{1,...,M\}} +U$ vanish. 
Notice that $d>0$ and $d<1$ since
$\gamma \sum_\alpha z_\alpha = 1 -(1-l)/A<1$, see (\ref{equilibrium_resources}).

$Q$ is the generator of a time-continuous Markov chain with positive out of diagonal entries and negative diagonal entries. All of the  real parts of its eigenvalues are non-positive. We next observe that $\pi = (\pi_\alpha)$ with
$\pi_\alpha = 1/z_\alpha$ is invariant for $Q$:
\begin{eqnarray*}
\sum_\beta \pi_\beta q_{\beta\alpha}&=&\sum_{\beta\ne\alpha}\pi_\beta\gamma w_\beta z_\alpha + \pi_\alpha (w_\alpha (d-1)+\gamma w_\alpha z_\alpha)\\
&=&\gamma z_\alpha \sum_{\beta\ne\alpha} z_\beta -\pi_\alpha \sum_{\beta\ne\alpha}q_{\alpha\beta}\\
&=& \gamma z_\alpha \sum_{\beta\ne\alpha} z_\beta -\pi_\alpha\sum_{\beta\ne\alpha}\gamma w_\alpha z_\beta =0.
\end{eqnarray*}
Consider the scalar product on $\R^M$ given by
$\langle f,g\rangle_\pi = \sum_\alpha f(\alpha)g(\alpha)\pi_\alpha$. The adjoint $Q^*$ of $Q$ with respect to this scalar product is
$$q_{\alpha\beta}^* =\frac{\pi_\beta q_{\beta\alpha}}{\pi_\alpha}.$$
From construction,
$$q_{\alpha\beta}^* = \frac{\frac{1}{z_\beta}\gamma w_\beta z_\alpha}{\frac{1}{z_\alpha}}= \gamma z_\beta w_\alpha = q_{\alpha\beta},$$
so that $Q=Q^*$ is self-adjoint, and therefore all of its eigenvalues are real non-positive, so that
$$\sup_{\vert\vert v\vert\vert_\pi =1} \langle v,Qv\rangle_\pi =0,$$
which yields that
$$d \sum_\alpha v_\alpha^2 w_\alpha  +
\langle v,Uv\rangle_\pi \le 0,\ \forall v\in\R^M,$$
implying that $\langle v,Uv\rangle_\pi < 0$
$\forall v\in\R^M$. Hence
$\langle v,  D_{\pi} U v\rangle <0$ $\forall v\in\R^M$, where $\langle \cdot,\cdot\rangle$ denotes the standard scalar product in $\R^M$
and $D_{\pi}=\diag(\pi_\alpha)_{\alpha \in \{1,...,M\}}$.
This implies that its symmetrized version $(D_{\pi}U +U^T D_{\pi})/2$ is stable, so that $U^T$ is Volterra-Lyapunov diagonally stable, and is thus D-stable, see, e.g. \parencite{cross}. $U^T$ is therefore stable. Thus $C^T U^T C$ is Volterra-Lyapunov diagonally stable and D-stable. This implies the stability of the equilibria since
$$J'(y^*) = ((1-l)\mu/A) \diag(g_iy_i^*)_{i\in \{1,...,S\}} \Lambda$$ with $\Lambda = C^T U^T C$.
\end{proof}

\begin{remark}
    $y^* = \vec{0}$ is always an equilibrium for the slow species o.d.e. system, and for all $\alpha$, $R^*(\vec{0}) = \frac{\mu_\alpha}{\nu_\alpha} $.
The Jacobian at equilibrium point $y^* = \vec{0}$ writes :
$$J_{i,j}(\vec{0})=g_i\delta_{i,j}\left((1-l)\sum_\alpha \frac{c_{i,\alpha}\mu_\alpha}{\nu_\alpha}-m_i \right)$$
This equilibrium is unstable if it exists $i \in \{1,...,S\}$ such that :
$$g_i(1-l)\sum\limits_\alpha\frac{c_{i,\alpha}\mu_\alpha}{\nu_\alpha} > m_i$$
\begin{remark}
    The leakage parameter $l$ does not affect equilibrium stability, but it may affect equilibrium abundances and species survival. In the setting studied numerically by Robert Marsland et al. \parencite{marsland19}, where only one resource is externally supplied ($\mu_1\neq 0$ and $\mu_i= 0$ for all $i\in\{2,...,S\}$), increasing $l$ was found to increase the number of surviving species through resources redistribution induced by the matrix $D$. By contrast, in our setting, where all resources are externally supplied, $l$ has only a negligible influence on the number of survivors.
\end{remark}


\end{remark}

\section{Stochastic slow-fast  system}
We exploit the fact that the o.d.e. (\ref{dynamic_resources}) or (\ref{Gen2B}) is the mass action kinetics of a first order chemical reaction network (CRN) that leads to stochastic dynamics, 
 see,  e.g. \parencite{gadgil} or \parencite{MazzaBenaim}. 
The convergence of the trajectories of such Markov processes toward mass action o.d.e. in large volume limits is well established, for example for density-dependent Markov chains, see, e.g. \parencite{kurtz}. \\

Taking inspiration from the slow-fast approach of MacArthur and the deterministic slow-fast approximation of Section \ref{ConsumerFast}, we consider a pair of stochastic processes
$\eta^g(t)=(R^g(t),Y^g(t))$, $t\ge 0$, taking values in $\mathbb{Z}_{\geq 0}^M\ {\rm x}\ \mathbb{Z}_{\geq 0}^S$ counting the number of resource molecules of each kind and the abundances of the $S$ species at time $t$.
The generator of this time continuous Markov chain is of the form $Q^g= A/g + B$, where $A$ gives the transition rates of the fast process for transitions
like $(r,y)\to (r',y)$
for fixed $y\in \mathbb{Z}_{\geq 0}^S$, and $B$ is the generator giving  the rates of transitions $(r,y) \to (r',y')$ from $\mathbb{Z}_{\geq 0}^M\ {\rm x}\ \mathbb{Z}_{\geq 0}^S$. For more details on such processes, see \parencite{Zhang1997}.
\subsection{Poisson equilibrium for  the fast stochastic system}\label{SCRN}

In our situation, for a given slow state $y\in\mathbb{Z}_{\geq 0}^S$, the fast process of generator $A$
is obtained by observing that 
(\ref{Gen2B}) is the mass action kinetics associate to the CRN given by the chemical reactions
\begin{align*}\label{eq: CRN}
\begin{aligned}
\ch{$\emptyset$ &->[ $\mu_\alpha$ ] $W_\alpha$},\\
\ch{$W_\alpha$ &->[ $\bar\nu_\alpha(y)$ ] $\emptyset$},\\
\ch{$W_\alpha$ &->[ $\kappa_{\alpha,\beta}(y)$ ] $W_\beta$}, \quad \alpha \neq \beta,
\end{aligned}
\end{align*}
where $W_\alpha$ is a symbol for resource $\alpha$  and where we set (see Section \ref{MatrixFormulation})
$$\bar\nu_\alpha(y)=\nu_\alpha +x_\alpha(y)(1-l P_{\alpha \beta}),$$
and
$$\kappa_{\alpha,\beta}(y)=l x_\alpha(y) P_{\alpha\beta},$$
for uptake rates $x_\alpha$ defined by $x_\alpha(y) = \sum_i c_{i\alpha} y_i$.\\

We adopt the standard terminology of CRN theory, see, e.g. \parencite{craciun,MazzaBenaim}.
This CRN is maybe not the only choice of all CRN giving the same dynamics but it is unique among those being a weakly reversible deficiency
zero network \parencite{craciun2021uniqueness}.
The species set of the CRN is ${\mathcal S}=\{ W_\alpha, \ \alpha=1,\ldots, M\}$, of size
$M$ which are also given by
the canonical basis vectors $e_\alpha$ of $\R^M$. The complexes set of this CRN is ${\mathcal C}=\{\emptyset, W_\alpha,\ \alpha=1,\ldots, M\}$, of size
$M+1$.   The set of reactions is given by the directed graph
$${\mathcal E}=\{(\emptyset, W_\alpha),\ (W_\alpha,\emptyset),\ (W_\alpha, W_\beta), \forall \alpha \neq \beta\},$$
when the related transition rates are positive.
The CRN is weakly reversible when, e.g., 
$L_l$ is irreducible, see (\ref{Laplacienl}). In this case, the CRN has a single linkage class, and its stoichiometric subspace is the span of the vectors
$e_\alpha$, so that its dimension equals $M$. The deficiency is thus 
$\delta =(M+1)- 1- M=0$. The deficiency zero Theorem (see \parencite{craciun})  states that
each compatibility class contains a unique equilibrium $R^*(y)$ with $F_y(R^*(y))=0$ (see Sections \ref{ConsumerFast} and \ref{ClosedForm}) and that
the related stochastic mass action kinetics has a unique steady state  distribution $\pi_y$
of product Poisson form
$$\pi_y(r)=\prod_{\alpha=1}^M \frac{(R_\alpha^*(y))^{r_\alpha}}{(r_\alpha)!} e^{-R_\alpha^*(y)},$$
where $r=(r_1,\ldots,r_M)^T\in\mathbb{Z}_{\geq 0}^M$ is any possible resource configuration, see, e.g., \parencite{craciun,gadgil}.

\subsection{The slow species Markov chain}
To define the slow Markov chain, we take inspiration from \parencite{Zhang1997}, where an aggregated process $\bar \eta^g(t)$  
is defined by setting $\bar\eta^g(t)= y\in\mathbb{Z}_{\geq 0}^S$ when $\eta^g(t)\in \{y\}\ {\rm x}\ \mathbb{Z}_{\geq 0}^M$. When $g\to 0$, the authors of \parencite{Zhang1997} proved (in the finite case) that the aggregated process converges toward a limiting Markov chain on the set of slow states from $\mathbb{Z}_{\geq 0}^S$.
Basically, transitions between slow states $y$ and $y'$ from $\mathbb{Z}_{\geq 0}^S$
are obtained by averaging transition rates of transition $(r,y)\to (r',y')$ when $r$ is assumed to follow the fast CRN steady-state distribution $\pi_y$.\\

We propose therefore a stochastic model where
the slow Markov chain of states $y =(y_i)_{i=1,\ldots,S}$ evolves in the slow state space
$E_s =\mathbb{Z}_{\geq 0}^S$. Given $y\in E_s $ and $i=1,\ldots,S$, define $\bar y_i$ to be the configuration $y$ without its $i$th component, so that
$y=(\bar y_i,y_i)$.
The transition rate  $q^s(y,y')$ of the slow process for $y\in E_s$
and $y'=(\bar y_i,y_i+1)$, is given by
\begin{eqnarray*}
q^s(y,y')&=& g_i(1-l)\sum_r \sum_\alpha \pi_y(r) c_{i\alpha}  y_i r_\alpha\\
      &=& g_i(1-l)y_i\sum_\alpha c_{i\alpha} \sum_r \pi_y(r) r_\alpha \\
      &=& g_i (1-l) y_i \sum_\alpha c_{i\alpha}\E_{\pi_y}(r_\alpha)\\
      &=& g_i(1-l) y_i \sum_\alpha c_{i,\alpha} R_\alpha^*(y),
\end{eqnarray*}
and for degradation when $y'=(\bar y_i,y_i-1)$,
$$q^s(y,y')=g_im_i y_i.$$

Each component $y_i(t)$ evolves as a birth and death process, with  transition rates that depend on the state $y(t)$ and on  the instantaneous equilibrium $R^*(y(t))$ with $F_{y(t)}(R^*(y(t)))=0$. $y(t)$ is a Multivariate Competition Process (MCP) \parencite{iglehart1964multivariate} 
of birth and death rates  given by
$$b_i(y) = g_iy_i(1-l) \sum_{\alpha}  c_{i \alpha} R_{\alpha}^*(y)\hbox{ and }
d_i(y)=g_i m_i y_i,$$
for $i
=1,\cdots,S$.
 Once a
coordinate $y_j(t)$ of the process hits 0 i.e. when species $j$ gets extinct, it remains equal to 0.
The subset $\Delta := \mathbb{Z}_{\geq 0}^S - \mathbb{Z}_{>0}^S$ is
absorbing for the process $y(t)$, and it is natural to look at the  quasi-stationary distribution, which is the distribution conditioned to non-extinction \parencite{meleard2012quasi} .

\subsubsection{Absorption to \texorpdfstring{$\vec{0}$}{Lg}}
Following \parencite{chazottes_time_2019}, consider the monotype birth and death process $\left(\sum_{i=1}^S y_i(t)\right)_{t\geq0}$. 

Using (\ref{bound_birth}) and the inequality $0<g_i<1$, this process is stochastically bounded by another monotype birth and death process with birth rate at $k \in \mathbb{Z}_{\geq 0}$ given by $b^*(k)=\sum_\alpha \mu_\alpha \equiv \phi$ and $d^*(k) = k \min_i g_im_i  $.


 Using the fact that $$\sum_{k=1}^\infty \frac{d^*(1)\dots d^*(k)}{b^*(1)\dots b^*(k)} = \sum_{k \geq 1} \left(\frac{\min_i m_i}{\phi}\right)^k k!$$ diverges, 
 Theorem 5.5.5. of \parencite{meleard2016modeles},
shows that the birth and death process $\left(\sum_{i=1}^S y_i(t)\right)_{t\geq0}$ go extinct almost surely, and therefore the multi-type birth and death process $y(t)$ is absorbed in  $\vec{0}$ with probability one.

Let
$$\tau_{\Delta} = \inf \{t \geq 0;  y(t) \in \Delta \} ,$$
be the absorption time of the process in $\Delta$. Results of  \parencite{iglehart1964multivariate} show that
the process $y(t)$ is  absorbed a.s. in $\Delta$ with finite mean absorption time i.e.
$$ \mathbb{E}(\tau_{\Delta}) < + \infty .$$

\subsubsection{Convergence to a quasi-stationary distribution (QSD)}

As  $y(t)$ is absorbed a.s. in $\Delta$ with finite mean absorption time, the process does not possess a stationary distribution.
The quasi-stationary distribution (QSD) is the law of the process conditioned to non-extinction. Let $E^*=\Z_{\geq 0}^S\setminus \Delta$. The following definition can be found. e.g., in  \parencite{meleard2012quasi}:
 \begin{definition}  
A probability measure $\alpha$ on $E^*$ is a quasi-stationary distribution for the CTMC $y(t)$, if $ \forall t \geq 0$ and any measurable set $A \subset E^*$:
\begin{align*}
    \alpha(A) = \mathbb{P}_{\alpha}(y_t \in A |\tau_{\Delta}<t)
\end{align*}
\end{definition}
A large literature on QSD has emerged recently with a focus on population processes see e.g. \parencite{meleard2012quasi,collet2013quasi}. If finding the QSD is a  difficult problem \parencite{benaim2021stochastic}, it is possible to find criteria on the transition rates ensuring  convergence of the law of the process toward  a QSD, see, e.g, \parencite{champagnat2023general}.
Some further conditions on the transition rates are needed to obtain uniform exponential convergence to a unique QSD \parencite{champagnat2021lyapunov}. We now show that the process $y(t)$ satisfies a sufficient condition ensuring the convergence of the law of the process toward  a QSD. For  multitype birth and death process $y(t)$,
Example 5.4 of \parencite{champagnat2023general} shows that a sufficient condition  ensuring the convergence is the existence of a constant $\theta > 1$ such that
\begin{equation}\label{condition_qss}
    \sum_{i=1}^S (d_i(y)- \theta b_i(y)) \xrightarrow[y \in E^*, |y|\to +\infty]{} +\infty 
\end{equation}  where $|y| = y_1+...+y_S$.
Let $\theta > 1$. Then
\begin{align*}
    \sum_{i=1}^S \left(d_i(y)- \theta b_i(y)\right) &= \sum_{i=1}^S g_i\left(y_i m_i- \theta  \sum_{\alpha} c_{i,\alpha} y_i R_{\alpha}\right) \\
    & \geq  \sum_{i=1}^S g_i y_i m_i - S \theta \phi\\
    & \geq |y| \min_{i \in \{1,...,S\}} g_im_i  - S \theta \phi  \xrightarrow[|y|\to +\infty]{}+\infty,
\end{align*}
and it follows that  (\ref{condition_qss}) is satisfied.  Theorem 3.5 of \parencite{champagnat2023general} ensures then that $y(t)$ admits a quasi-stationary distribution. However, the sufficient conditions given in \parencite{champagnat2021lyapunov}  for uniform exponential convergence to a unique QSD are not satisfied.

\begin{remark}
    Note that here by considering absorption to $\Delta$, we have an irreducible state space $E^* = \mathbb{Z}_{\geq 0}^S-\Delta$ which is not the case if we consider absorption in $\vec{0}$ and $E^* = \mathbb{Z}_{\geq 0}^S-\vec{0}$.
\end{remark} 

\section{Conclusion}
Microbial communities play a vital role in ecosystems and human health, yet their complex dynamics—including competition and cross-feeding—are poorly captured by traditional models. This study advances our understanding by analyzing the Microbial Consumer-Resource Model (MiCRM), which explicitly incorporates metabolite exchanges. Using a slow-fast approximation, we demonstrate that microbial communities are unlikely to persist when consumer species outnumber resources, while balanced communities remain stable. Extending the model stochastically also reveals how variability can drive extinction events.
Even if no model captures nature perfectly, this work provides a more realistic framework for studying microbial communities, paving the way for further applications.

\section{Funding}
This work was supported in part by the Swiss National Science Foundation (Sinergia program, grant CRSII5 189919/1 to J.M. and C.M.).

\printbibliography

@article{marsland19,
  title={Available energy flux drive a transition in the diversity, stability, and functional structure of microbial communities},
  author={Marsland III, Robert and Cui, Wenping and Goldford, Joshua and Sanchez, Alvaro  and Mehta, Pankaj},
  journal={PLOS Comput. Biol.},
  year={2019},
  doi={https://doi.org/10.1371/journal.pcbi.1006793 }
  }

@article{mustri,
  title={Accuracy of the Lotka-Volterra model fails in strongly coupled microbial consumer-resource systems},
  author={Mustri, Michael P and Duan, Quqiming and Pawar, Samraat},
  journal={PLOS Computational Biology},
  volume={21},
  number={12},
  pages={e1013719},
  year={2025},
  publisher={Public Library of Science San Francisco, CA USA},
  doi={https://doi.org/10.1371/journal.pcbi.1013719 }
}

@article{momeni,
  title={Lotka-Volterra pairwise modeling fails to capture diverse pairwise microbial interactions},
  author={Momeni, Babak and Xie, Li and Shou, Wenying},
  journal={elife},
  volume={6},
  pages={e25051},
  year={2017},
  publisher={eLife Sciences Publications, Ltd},
  doi={https://doi.org/10.7554/elife.25051 }
}

@article{monod,
  title={The Growth of Bacterial Cultures},
  author={Monod, Jacques},
  journal={Annual Review on Microbiology},
  volume={3},
  pages={371--394},
  year={1949},
  doi={https://doi.org/10.1016/b978-0-12-460482-7.50020-8 }
  }

@article{van2022ecological,
  title={Ecological modelling approaches for predicting emergent properties in microbial communities},
  author={van den Berg, Naomi Iris and Machado, Daniel and Santos, Sophia and Rocha, Isabel and Chac{\'o}n, Jeremy and Harcombe, William and Mitri, Sara and Patil, Kiran R},
  journal={Nature Ecology \& Evolution},
  pages={1--11},
  year={2022},
  publisher={Nature Publishing Group},
  doi={https://doi.org/10.1038/s41559-022-01746-7 }
}

@article{culp2023cross,
  title={Cross-feeding in the gut microbiome: ecology and mechanisms},
  author={Culp, Elizabeth J and Goodman, Andrew L},
  journal={Cell Host \& Microbe},
  volume={31},
  number={4},
  pages={485--499},
  year={2023},
  publisher={Elsevier},
  doi={https://doi.org/10.1016/j.chom.2023.03.016 }
}

@article{lobry,
  title={On Tykhonovs Theorem for convergence of solutions of slow and fast systems},
  author={Lobry, C. and Sari, T. and Touhami, S.},
  journal={Electronic Journal of Differential Equations},
  volume={19},
  pages={1--22},
  year={1998},
  }

@article{craciun,
  title={Product-Form Stationary Distributions for Deficiency Zero Chemical Reaction Networks},
  author={Anderson, D. and Craciun, G. and Kurtz, T.},
  journal={Bul. Math. Biol.},
  volume={72},
  pages={1947--1970},
  year={2010},
  doi={https://doi.org/10.1007/s11538-010-9517-4 }
  }

@article{sherman1950adjustment,
  title={Adjustment of an inverse matrix corresponding to a change in one element of a given matrix},
  author={Sherman, Jack and Morrison, Winifred J},
  journal={The Annals of Mathematical Statistics},
  volume={21},
  number={1},
  pages={124--127},
  year={1950},
  publisher={JSTOR},
  doi={https://doi.org/10.1214/aoms/1177729893 }
}

@article{collet2013quasi,
  title={Quasi-stationary distributions},
  author={Collet, Pierre and Mart{\'\i}nez, Servet and San Mart{\'\i}n, Jaime},
  journal={Markov chains, diffusions and dynamical systems},
  year={2013},
  doi={ https://doi.org/10.1007/978-3-642-33131-2 }
}

@article{altenberg,
  title={Karlin theory on growth and mixing extended to linear differential equations},
  author={Altenberg, Lee},
  journal={arXiv preprint arXiv:1006.3147},
  year={2010}
}

@article{chen2022two,
  title={Two novel proofs of spectral monotonicity of perturbed essentially nonnegative matrices with applications in population dynamics},
  author={Chen, Shanshan and Shi, Junping and Shuai, Zhisheng and Wu, Yixiang},
  journal={SIAM Journal on Applied Mathematics},
  volume={82},
  number={2},
  pages={654--676},
  year={2022},
  publisher={SIAM},
  doi={https://doi.org/10.1137/20m1345220 }
}

@article{chen2023impact,
  title={On the impact of spatial heterogeneity and drift rate in a three-patch two-species Lotka--Volterra competition model over a stream},
  author={Chen, Shanshan and Liu, Jie and Wu, Yixiang},
  journal={Zeitschrift f{\"u}r angewandte Mathematik und Physik},
  volume={74},
  number={3},
  pages={117},
  year={2023},
  publisher={Springer},
  doi={https://doi.org/10.1007/s00033-023-02009-6 }
}

@book{meleard2016modeles,
	series = {Mathématiques et {Applications}},
	title = {Modèles aléatoires en {Ecologie} et {Evolution}},
	volume = {77},
	copyright = {http://www.springer.com/tdm},
	isbn = {9783662494547 9783662494554},
	url = {http://link.springer.com/10.1007/978-3-662-49455-4},
	language = {fr},
	urldate = {2026-03-26},
	publisher = {Springer Berlin Heidelberg},
	author = {Méléard, Sylvie},
	year = {2016},
	doi = {10.1007/978-3-662-49455-4},
}

@article{mehta2021cross,
  title={Cross-feeding shapes both competition and cooperation in microbial ecosystems},
  author={Mehta, Pankaj and Marsland III, Robert},
  journal={arXiv preprint arXiv:2110.04965},
  year={2021},
  doi={https://doi.org/10.1101/2021.10.10.463852 }
}

@article{meleard2012quasi,
	title = {Quasi-stationary distributions and population processes},
	doi = {10.1214/11-PS191},
	urldate = {2026-03-26},
	journal = {Probability Surveys},
	author = {Méléard, Sylvie and Villemonais, Denis},
	year = {2012},
    doi={https://doi.org/10.1214/11-ps191 }
}

@article{champagnat2021lyapunov,
  title={Lyapunov criteria for uniform convergence of conditional distributions of absorbed Markov processes},
  author={Champagnat, Nicolas and Villemonais, Denis},
  journal={Stochastic Processes and their Applications},
  volume={135},
  pages={51--74},
  year={2021},
  publisher={Elsevier},
  doi={https://doi.org/10.1016/j.spa.2020.12.005 }
}

@article{champagnat2023general,
  title={General criteria for the study of quasi-stationarity},
  author={Champagnat, Nicolas and Villemonais, Denis},
  journal={Electronic Journal of Probability},
  volume={28},
  pages={1--84},
  year={2023},
  publisher={The Institute of Mathematical Statistics and the Bernoulli Society},
  doi={https://doi.org/10.1214/22-ejp880 }
}

@article{poggiale2020analysis,
  title={Analysis of a predator--prey model with specific time scales: a geometrical approach proving the occurrence of canard solutions},
  author={Poggiale, Jean-Christophe and Aldebert, Clement and Girardot, Benjamin and Kooi, Bob W},
  journal={Journal of mathematical biology},
  volume={80},
  number={1},
  pages={39--60},
  year={2020},
  publisher={Springer},
  doi={https://doi.org/10.1007/s00285-019-01337-4 }
}

@article{dougoud,
  title={The feasibility of equilibria in large ecosystems: a primary but neglected concept in the complexity-stability debate},
  author={Dougoud, M. and Vinckenbosch, L. and Rohr, RP. and Bersier, L-F. and Mazza, C.},
  journal={PLOS Computational Biology},
  volume={14},
  number={2},
  year={2018},
  doi={https://doi.org/10.1371/journal.pcbi.1005988 }
  }

@article{jamal2,
  title={Complex systems in ecology: a guided tour with large Lotka--Volterra models and random matrices},
  author={Akjouj, Imane and Barbier, Matthieu and Clenet, Maxime and Hachem, Walid and Maida, Mylene and Massol, Francois and Najim, Jamal and Tran, Viet Chi},
  journal={Proceedings of the Royal Society A: Mathematical, Physical and Engineering Sciences},
  volume={480},
  number={2285},
  year={2024},
  publisher={The Royal Society},
  doi={https://doi.org/10.1098/rspa.2023.0284 }
}

@article{jamal1,
  title={Positive Solutions for Large Random Linear Systems},
  author={Bizeul, P. and Najim, J.},
  journal={Proc. AMS},
  volume={149},
   number={6},
  pages={2333--2348},
  year={2021},
  doi={https://doi.org/10.1109/icassp40776.2020.9053593 }
  }

@article{stone,
  title={The Google matrix controls the stability of structured ecological and biological networks},
  author={Stone, Lewi},
  journal={Nature communications},
  volume={7},
  number={1},
  pages={12857},
  year={2016},
  publisher={Nature Publishing Group UK London},
  doi={https://doi.org/10.1038/ncomms12857 }
}

@article{Guex,
  title={Regulated bacterial interaction network: A mathematical framework to describe growth under inclusion of metabolite cross-feeding},
  author={Guex, Isaline and Mazza, Christian and Dubey, Manupriam and Li, Renyi and Batsch, Maxime and van der Meer, Jan Roelof},
  journal={PLOS Computational Biology},
  year={2023},
  doi={https://doi.org/10.1371/journal.pcbi.1011402 }
  }

@article{clegg,
  title={Variation in thermal physiology can drive the temperature-dependence of microbial communities richness},
  author={Clegg, T. and Pawar, S.},
  journal={Elife},
   volume={13},
  year={2024},
  doi={https://doi.org/10.7554/elife.84662 }
  }

@article{Marsland2019,
  title={The minimum environmental perturbation principle: A new perspective on niche theory},
  author={Marsland III, Robert and Cui, Wenping and Mehta, Pankaj},
  journal={The American Naturalist},
  volume={196},
  number={3},
  pages={291--305},
  year={2020},
  publisher={The University of Chicago Press Chicago, IL},
  doi={https://doi.org/10.1086/710093 }
}

@book{kurtz,
	title={Markov Processes: Characterization and Convergence},
    series = {Wiley series in probability and mathematical statistics},
	author={Ethier, S. and Kurtz, T.},
	year={1986},
	publisher={Wiley, New-York}
}

@article{Zhang1997,
  title={Structural properties of Markov chains with weak and strong interactions},
  author={Zhang, Q. and Yin, G.},
  journal={Stoch. Proc. Appl.},
  pages={181--197},
  year={1997},
  publisher={Elsevier},
  doi={https://doi.org/10.1007/978-1-4612-0627-9_7 }
}

@article{marsland2020community,
  title={The Community Simulator: A Python package for microbial ecology},
  author={Marsland, Robert and Cui, Wenping and Goldford, Joshua and Mehta, Pankaj},
  journal={Plos one},
  volume={15},
  number={3},
  pages={e0230430},
  year={2020},
  publisher={Public Library of Science San Francisco, CA USA},
  doi={https://doi.org/10.1371/journal.pone.0230430 }
}

@article{marsland2020minimal,
  title={A minimal model for microbial biodiversity can reproduce experimentally observed ecological patterns},
  author={Marsland III, Robert and Cui, Wenping and Mehta, Pankaj},
  journal={Scientific reports},
  volume={10},
  number={1},
  pages={3308},
  year={2020},
  publisher={Nature Publishing Group UK London},
  doi={https://doi.org/10.1038/s41598-020-60130-2 }
}

@article{MacArthur69,
  title={Species packing and what competition minimizes},
  author={MacArthur, R.H.},
  journal={Proceedings of the National Academy of Sciences},
  pages={1369--1371},
  year={1969},
  publisher={},
  doi={https://doi.org/10.1073/pnas.64.4.1369 }
}

@article{MacArthur70,
  title={Species packing and Competitive Equilibrium for Many Species},
  author={MacArthur, R.H.},
  journal={Theoretical Population Biology},
  pages={1--11},
  year={1970},
  publisher={},
  doi={https://doi.org/10.1016/0040-5809(70)90039-0 }
}

@article{gadgil,
	title={A stochastic analysis of first-order reaction networks},
	author={Gadgil, C. and Lee, C. and Othmer, H.},
	journal={Bull. Math. Biol.},
	volume={67},
	pages={901--946},
	year={2005},
	publisher={Springer},
    doi={https://doi.org/10.1016/j.bulm.2004.09.009 }
}

@misc{mischler,
  author        = {Stéphane Mischler},
  title         = {An introduction to evolution PDEs},
  year          = {2019},
howpublished = {Lecture notes, Université Paris Dauphine}
}

@article{benaim2021stochastic,
	title = {Stochastic approximation of quasi-stationary distributions for diffusion processes in a bounded domain},
	volume = {57},
	issn = {0246-0203},
	doi = {https://doi.org/10.1214/20-aihp1093 },
	number = {2},

	journal = {Annales de l'Institut Henri Poincaré, Probabilités et Statistiques},
	author = {Benaïm, Michel and Champagnat, Nicolas and Villemonais, Denis},
	year = {2021},
}

@book{MazzaBenaim,
  title={{Stochastic Dynamics for Systems Biology}},
  author= {Christian Mazza and Michel Benaim},
  Series = {Mathematical and Computational Biology Series},
  year={2014},
  publisher={CRC Press}	
  }

@book{stoer1980introduction,
  title={Introduction to numerical analysis},
  author={Stoer, Josef and Bulirsch, Roland and Bartels, R and Gautschi, Walter and Witzgall, Christoph},
  volume={1993},
  year={1980},
  publisher={Springer}
}

@article{cross,
title={Three Types of Matrix Stability},
author={Cross, G.},
journal={Lin. Alg. Appl.},
  volume={20},
  pages={253--263},
  year={1978},
  doi={https://doi.org/10.1016/0024-3795(78)90021-6 }
}

@article{chesson1990macarthur,
  title={MacArthur's consumer-resource model},
  author={Chesson, Peter},
  journal={Theoretical Population Biology},
  volume={37},
  number={1},
  pages={26--38},
  year={1990},
  publisher={Elsevier},
  doi={ https://doi.org/10.1016/0040-5809(90)90025-q }
}

@article{craciun2021uniqueness,
  title={Uniqueness of weakly reversible and deficiency zero realizations of dynamical systems},
  author={Craciun, Gheorghe and Jin, Jiaxin and Polly, Y Yu},
  journal={Mathematical biosciences},
  volume={342},
  pages={108720},
  year={2021},
  publisher={Elsevier},
  doi={https://doi.org/10.1016/j.mbs.2021.108720 }
}

@article{baron2023breakdown,
  title={Breakdown of random-matrix universality in persistent Lotka-Volterra communities},
  author={Baron, Joseph W and Jewell, Thomas Jun and Ryder, Christopher and Galla, Tobias},
  journal={Physical Review Letters},
  volume={130},
  number={13},
  pages={137401},
  year={2023},
  publisher={APS},
  doi={https://doi.org/10.1103/physrevlett.130.137401 }
}

@article{iglehart1964multivariate,
  title={Multivariate competition processes},
  author={Iglehart, Donald L},
  journal={The Annals of Mathematical Statistics},
  pages={350--361},
  year={1964},
  publisher={JSTOR},
  doi={https://doi.org/10.1214/aoms/1177703758 }
}

@article{chazottes_time_2019,
	title = {On time scales and quasi-stationary distributions for multitype birth-and-death processes},
	volume = {55},
	issn = {0246-0203},
	number = {4},
	journal = {Annales de l'Institut Henri Poincaré, Probabilités et Statistiques},
	author = {Chazottes, J.-R. and Collet, P. and Méléard, S.},
	year = {2019},
    doi={https://doi.org/10.1214/18-aihp948 }
}
\end{document}